\documentclass[journal,draftcls,onecolumn,12pt,twoside]{IEEEtranTCOM}
\normalsize

\usepackage{multirow}
\usepackage{graphicx}
\usepackage{amsmath, amsfonts}
\usepackage[linesnumbered,ruled]{algorithm2e}
\usepackage{algpseudocode}
\usepackage{subfigure}
\usepackage{url}
\usepackage{xcolor}

\newtheorem{lemma}{Lemma}

\newtheorem{theorem}{Theorem}

\begin{document}

\title{On the Exact Lower Bounds of Encoding Circuit Sizes of Hamming Codes and Hadamard Codes
}


\author{Zhengrui Li,   Sian-Jheng Lin,~\IEEEmembership{Member,~IEEE,}   Yunghsiang S. Han,~\IEEEmembership{Fellow,~IEEE,}
\thanks{This work was partially funded by the CAS Pioneer Hundred Talents Program under contract no. KJ2100060007 and by the NSFC of China (No.~61671007).
Z. Li and S.-J. Lin are with
the CAS Key Laboratory of Electro-magnetic Space Information, the School of Information Science and Technology, University of Science and Technology of China (USTC), Hefei,Anhui, China, email: dd622089@mail.ustc.edu.cn, sjlin@ustc.edu.cn}            
\thanks{Y. S. Han is with the
School of Electrical Engineering \& Intelligentization,
Dongguan University of Technology, email: yunghsiangh@gmail.com
}
}




\maketitle

\begin{abstract}
In this paper, we investigate the encoding circuit size of Hamming codes and Hadamard codes. To begin with, we prove the exact lower bound of circuit size required in the encoding of (punctured)~Hadamard codes and (extended)~Hamming codes. Then the encoding algorithms for (punctured)~Hadamard codes are presented to achieve the derived lower bounds. For (extended)~Hamming codes, we  also propose encoding algorithms that achieve the lower bounds.
\end{abstract}

\section{Introduction}
In the late 1940s, Claude Shannon founded a great theory of reliable communication in the presence of noise, which leads to the study of error-correcting codes~\cite{Hill,MS,2004-lin}, that are used to improve the error correction capability under certain coding rates. Later, the complexity of error-correcting codes raises the interest of theoretical researchers. The complexity includes complexity of operation circuits and control circuits, (intermediate) memory/storage requirement and encoding latency, and so on. In this paper, we consider the circuit sizes of the encoder of linear block codes. That is, we focus on the number of Boolean operations used in the encoding algorithms of certain linear block codes.

For an $[n,k,d]$ linear block code, the codeword is defined as $\mathbf{x}\mathbf{G}$, where $\mathbf{G}$ is the generator matrix, and $\mathbf{x}$ is the vector of $k$ message symbols. Thus, a trivial encoding algorithm requires $\mathcal{O}(kn)$ arithmetic operations. In addition, the encoding algorithm requires at least $\Omega(n)$ operations. It is shown in \cite{Barg}, given an $[n,k,d]$ binary linear code $\mathcal{C}$, we can construct another $[n,k,d]$ binary linear code $\mathcal{C}'$, whose encoding algorithm requires at most $kd$ exclusive-ORs~(XORs). Further, as there exists a class of binary linear codes reaching the Gilbert-Varshamov bound \cite{MS}, we conclude that there exists a class of binary linear codes reaching the Gilbert-Varshamov bound, and its encoding algorithm requires at most $d$ XORs per message bit. However, given any $[n,k,d]$ binary linear code, it is still unknown whether this code can be encoded using at most $d$ XORs per message bit.

One of the main issues of complexity theory is the complexity lower bound. However, it is difficult to explore the lower bound of the encoding complexities of most linear block codes. To date, a number of linear block codes with linear-time encoder $\mathcal{O}(n)$ are proposed, such as fountain codes~\cite{fountain,fountain1} and expander codes~\cite{556667,556668}. To our knowledge, the expander code is the only known asymptotically good linear block code to have linear-time encoders and decoders. However, for these linear-time codes, the constant factors hidden in the big-O complexity are unknown. Lin~\cite{SJ} shows that the MDS codes with two or three parities require at least $2+\epsilon$ XORs per message bit in encoding, where $\epsilon \rightarrow 0$. However, the factor $2+\epsilon$ is not tight. In this paper, we first give a tight lower bound of the encoding circuit sizes for two classes of  linear block codes that are the Hamming codes and the Hadamard codes.

In 1950, Richard Hamming~\cite{6772729} invented a class of error-correcting codes, termed as $[7,4]$ Hamming codes. Historically, Hamming codes are the first nontrivial family of error-correcting codes~\cite{720549}. This leads to a series of linear block codes, such as Reed-Muller codes and Reed-Solomon codes. The minimum Hamming  distance of Hamming codes is three, and this provides single-error correction or two-error detection. It is known that Hamming codes are perfects codes, which means that the Hamming codes achieve the optimal coding rate of the block codes with the minimum Hamming distance three~\cite{6772729}. Hence, Hamming codes have been implemented in systems which require faster and more accurate information transfers~\cite{6154204}. The extended Hamming codes are constructed by appending one extra parity-check bit to the Hamming codes~\cite{Hill,MS,6772729}. This extension allows the codes to perform sing-error correction and double-error detection simultaneously. Today, Hamming codes are commonly used to correct errors appeared in DRAM and SRAM~\cite{2004-lin,SRAM,RSHAM}. In addition, there are many codes that use the (extended) Hamming codes as the base codes~\cite{RSHAM,LDPCHAM,LDPCHAMM,LDPCHAMMM}.

Hadamard codes~\cite{2004-lin} are the dual codes of the Hamming codes. The Hadamard code is an error-correcting code that is used for error detections and corrections when transmitting messages are over extremely noisy or unreliable channels. 
An application of Hadamard codes is in code-division multiple-access (CDMA) systems~\cite{289411,295356,54460}. Hadamard codes are also used in communication systems for synchronizations and bandwidth spreading. In addition, the codes are also used in helicopter satellite communications \cite{had,hada}.
 
In this paper, we investigate the encoding computational complexities of Hamming codes and Hadamard codes. Precisely, we derive a exact lower bound on the circuit size of the encoding of (extended) Hamming codes and (punctured) Hadamard codes. Then encoding algorithms of these codes are proposed to achieve this bound. To our knowledge, we are the first to show a exact and tight lower bound on encoding computational complexity of non-trivial linear block codes.

In the rest of this paper, the background is given in Section~\ref{sec:2}. Section~\ref{sec:5} proves the exact lower bounds on encoding circuit size of these codes. Sections~\ref{sec:4.0} includes the algorithms and the circuit size analysis for Hadamard codes and punctured Hadamard codes, and the algorithms and analysis for the Hamming codes and the extended Hamming codes are given in Sections~\ref{sec:3}. It is shown in Sections~\ref{sec:4.0} and \ref{sec:3} that the lower bounds derived in Section~\ref{sec:5} are achievable. Section~\ref{sec:6} concludes this work.

\section{Preliminaries}\label{sec:2}
\begin{table}
\begin{center}
\caption{\label{tab:1}Table of notation}
\begin{tabular}{l|c}
\hline\hline
Term & Meaning\\
\hline
$[n]$ & $\{1,\dots,n\}$\\
\hline
$[m,n]$ & $\{m, m+1,\dots,n\}$\\
\hline
$\mathbf{G}_k$ & Generator matrix of $[2^k, k]$ Hadamard codes\\
\hline
$\mathbf{G}_k'$ & Generator matrix of $[2^k, k+1]$ punctured Hadamard codes\\
\hline
$\mathbf{G}_k''$ & Systematic generator matrix of $[2^k, k+1]$ punctured Hadamard codes\\
\hline
$\mathbf{H}_k$ & Parity-check matrix of $[2^k-1, 2^k-k-1]$ Hamming codes\\
\hline
$\mathbf{H}_k'$ & Parity-check matrix of $[2^k, 2^k-k-1]$ extended Hamming codes\\
\hline
$\mathbf{H}_k''$ & Systematic parity-check matrix of $[2^k, 2^k-k-1]$ extended Hamming codes\\
\hline 
$\mathbf{I}_{k}$ & Identity matrix of size $k\times k$\\
\hline 
$\mathbf{0}_n$ & All-Zeros row vector of length $n$\\
\hline
$\mathbf{1}_n$ & All-ones row vector of length $n$\\
\hline
$T^{i}(x)$& The  binary representation of  $x \in [0,2^{i}-1]$ with length $i$\\
\hline\hline
\end{tabular}
\end{center}
\end{table}

This section introduces the notations used in this paper, as well as the definitions of Hamming codes and Hadamard codes. For an $[n, k]$ linear code, $n$ denotes the code length and $k$ denotes the message (information) length. Table~\ref{tab:1} tabulates the common notations used in the paper. The function $T^{i}$ in Table~\ref{tab:1} is defined as follows. 
\begin{align*}
T^{i}:[0,2^{i}-1] &\to \mathbb {F}_2^i
\\a&\mapsto  [a_{i-1}\quad a_{i-2}\quad \dots\quad a_0]^T,
\end{align*}
where 
\[
a=a_{i-1}\cdot 2^{i-1}+a_{i-2}\cdot 2^{i-2}+\dots +a_0 
\]
and $A^T$ is the transpose of matrix $A$.
That is, $T^{i}(x)$ returns the binary representation of  $x$ with length $i$. It is clear that
\begin{equation}\label{eq:Tki}
T^{i}(x)+{\underbrace{[1\quad 0\quad \dots\quad 0]}_{i}}^T=T^{i}(x+2^{i-1})\qquad\forall x\in [0, 2^{i-1}-1].
\end{equation}

Furthermore, the average number of XORs per message bit is defined as
\[
\text{Avg.\ XORs}:=\lim_{\text{Size of the message} \rightarrow \infty}\frac{\text{Total number of XORs}}{\text{Size of the message}}.
\]
Throughout the paper, all arithmetic like multiplication, addition are in binary.

\subsection{Logical circuit}\label{sec:2A}
In this paper, we use the terminology logic circuit \cite{Barg}, or circuit for short, to represent the computational complexity. This is a directed tree whose nodes (gates) corresponding to Boolean operations from a chosen basis. An input string (bit) is fed to the input nodes of the circuit and the computational result appears on the output nodes. Complexity is measured by the number of gates (size) and the length of the longest path from an input to an output (parallel depth). For example, a circuit, to calculate $y_{0}=x_{0}+x_{1}+x_{2}$, is presented in Figure \ref{fig:fig4} with the input $\{x_{i}\}_{i=0}^{2}$ and the output $y_{0}$. There are only two addition gates $z_{0}$ and $y_{0}$ in Figure \ref{fig:fig4}, and therefore the circuit size is $2$. Further, the longest path from an input to an output is $x_{0}\rightarrow z_{0}\rightarrow y_{0}$, and thus the parallel depth is $2$.

\begin{figure}
	\center
	\includegraphics[width=0.5\columnwidth]{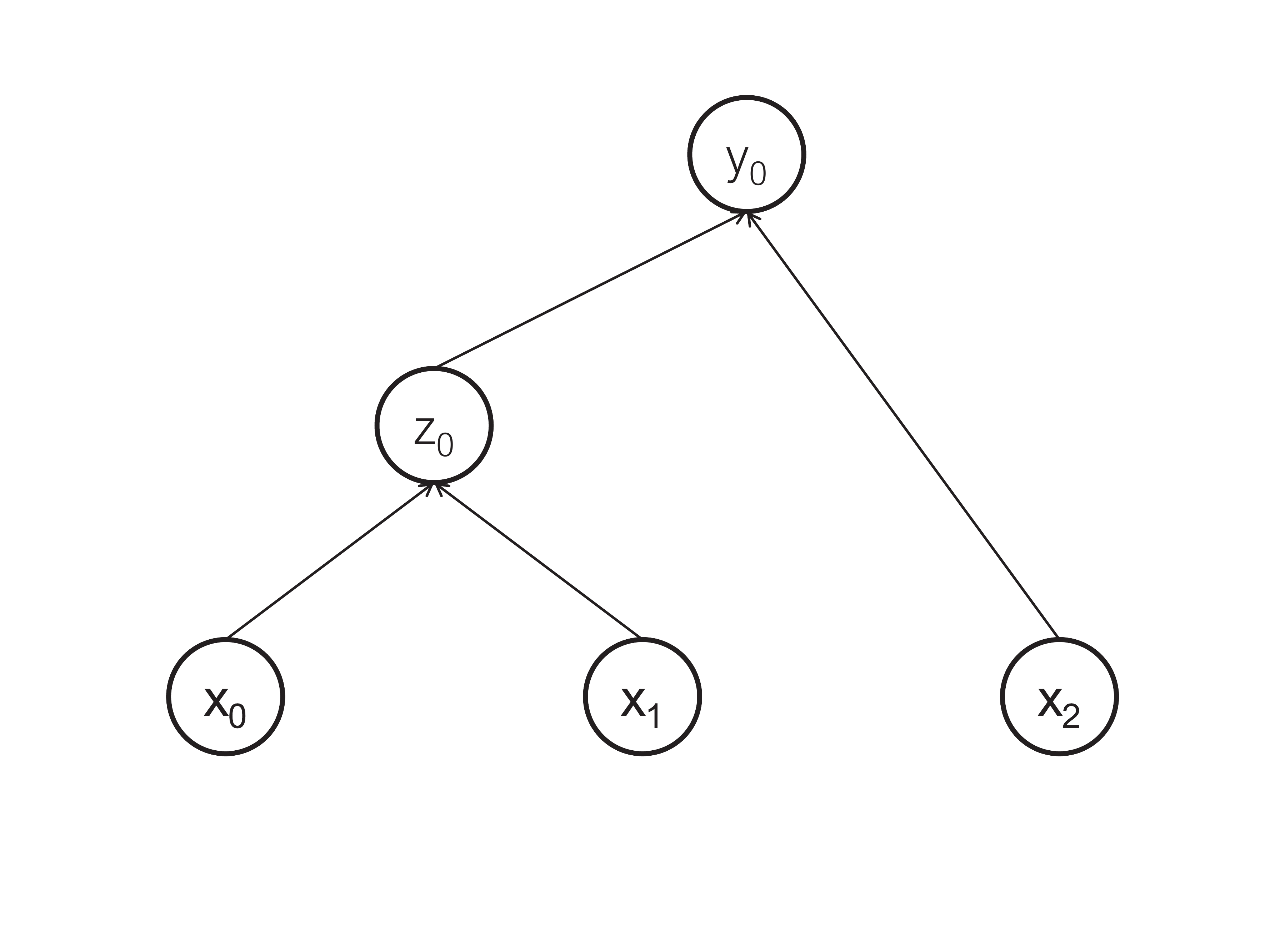}
	\caption{\label{fig:fig4}Circuit for $[4, 3]$ punctured Hadamard codes}
\end{figure}

\subsection{Systematic binary linear codes}
For an $[n=k+t, k]$ systematic binary linear code $C$, a codeword of $C$ consists of $k$ message symbols and $t$ parities over $\mathbb{F}_2$. W.L.O.G, the message symbols are located in the first $k$ symbols of the codeword, and the generator matrix is
\begin{equation}\label{eq:G0}
\mathbf{G}:=\begin{bmatrix}
\mathbf{I}_k & \mathbf{A}^T
\end{bmatrix},
\end{equation}
where $\mathbf{A}$ is a $t\times k$ matrix. Given a $k$-element message vector $\mathbf{x}$, the codeword $\mathbf{y}$ is defined as
\[
\mathbf{y}=\mathbf{x}\mathbf{G}=\begin{bmatrix}\mathbf{x}& \mathbf{x}\mathbf{A}^T\end{bmatrix}=\begin{bmatrix}\mathbf{x}& \mathbf{p}\end{bmatrix}.
\]
for any codeword $\mathbf{y}$, we have
\[
\mathbf{H}\mathbf{y}=\mathbf{0},
\]
where the parity check matrix is
\begin{equation}\label{eq:33}
\mathbf{H}:=\begin{bmatrix}
\mathbf{A} & \mathbf{I}_{t}
\end{bmatrix}.
\end{equation}

The dual code of $C$ is an $[n=k+t, t]$ systematic code with the generator matrix
\begin{equation}\label{eq:Gt}
\mathbf{G}^\perp=\begin{bmatrix} \mathbf{A}& \mathbf{I}_t\end{bmatrix}.
\end{equation}
Given a $t$-element message vector $\mathbf{x}^\perp$, the codeword $\mathbf{y}^\perp$ is defined as 
\[
\mathbf{y}^\perp=\mathbf{x}^\perp \mathbf{G}^\perp=\begin{bmatrix}\mathbf{x}^\perp \mathbf{A}& \mathbf{x}^\perp\end{bmatrix}=\begin{bmatrix}\mathbf{p}^\perp&\mathbf{x}^\perp\end{bmatrix}.
\]
The parity-check matrix for dual codes is the generator matrix for the original codes, and vice versa. 

\subsection{Hadamard codes}\label{sec:hadamard}
The generator matrix of $[2^k, k]$ Hadamard codes~\cite{4233219} with the minimum distance $2^{k-1}$ is given as follows. Let $\mathbf{G}_1:=\begin{bmatrix}0&1\end{bmatrix}$. For $k\geq 1$, a $k\times 2^{k}$ matrix $\mathbf{G}_k$ consists of all $k$-bit binary numbers in ascending  order. Precisely,  
\begin{equation}\label{eq:recur}
\begin{aligned}
\mathbf{G}_k:=
\begin{bmatrix}
\mathbf{0}_{2^{k-1}}& \mathbf{1}_{2^{k-1}}\\
\mathbf{G}_{k-1}&\mathbf{G}_{k-1}\\
\end{bmatrix}
= 
\begin{bmatrix}
T^k(0)& T^k(1)&\dots &T^k(2^k-1)
\end{bmatrix}.
\end{aligned}
\end{equation}
Notably, $\mathbf{G}_k$ has a zero column, which will generate a null codeword symbol. By removing this null symbol, the code is isomorphic to the $[2^k-1, k]$ binary Simplex codes~\cite{MS}.

The punctured Hadamard code is a variant of Hadamard codes. The generator matrix of $[2^k, k+1]$ punctured Hadamard codes is obtained by appending an all-ones row to $\mathbf{G}_k$. That is,
\begin{equation}\label{eq:GG}
\mathbf{G}'_k:=
\begin{bmatrix}
\mathbf{G}_k\\
\mathbf{1}_{2^k}\\
\end{bmatrix}.
\end{equation}
For example, the generator matrix of $[2^3, 4]$ punctured Hadamard codes is 
\begin{equation}\label{eq:G3}
\mathbf{G}'_3=
\begin{bmatrix}
\mathbf{G}_3\\
\mathbf{1}_8\\
\end{bmatrix}=
\begin{bmatrix}
0&0&0&0&1&1&1&1\\
0&0&1&1&0&0&1&1\\
0&1&0&1&0&1&0&1\\
1&1&1&1&1&1&1&1\\
\end{bmatrix}.
\end{equation}

However, \eqref{eq:GG} is not in the systematic form. To obtain the systematic form, the generator matrix is defined as
\begin{equation}\label{eq:G''_k}
\mathbf{G}''_k:=\mathbf{V}_{k+1} \mathbf{G}'_k,
\end{equation}
where 
\begin{equation}\label{eq:Vk+1}
\mathbf{V}_{k+1} :=
\begin{bmatrix}
\mathbf{I}_{k}&\mathbf{0}_k^T \\
\mathbf{1}_k & 1
\end{bmatrix}
\end{equation}
is a $(k+1)\times (k+1)$ binary non-singular matrix. For example, $\mathbf{G}''_3$ is
\begin{equation}\label{eq:GGs}
\mathbf{G}''_3=\mathbf{V}_4\cdot \mathbf{G}'_3=\begin{bmatrix}
0&0&0&0&1&1&1&1\\
0&0&1&1&0&0&1&1\\
0&1&0&1&0&1&0&1\\
1&0&0&1&0&1&1&0\\
\end{bmatrix},
\end{equation}
and
\begin{equation}\label{eq:V4}
\mathbf{V}_4=\begin{bmatrix}
1&0&0&0\\
0&1&0&0\\
0&0&1&0\\
1&1&1&1\\
\end{bmatrix}.
\end{equation}
From \eqref{eq:G''_k}, $\mathbf{G}''_k$ is obtained by replacing the $(k+1)$-th row of $\mathbf{G}'_k$ by the summation of all rows of $\mathbf{G}'_k$, and this leads that the Hamming weights of all $2^k$ columns of $\mathbf{G}''_k$ are odd. 

Let
\begin{equation}\label{eq:Ek}
\begin{aligned}
\mathbf{E}_{0}&:=[1],\\
\mathbf{E}_{k}&:=\begin{bmatrix}
\mathbf{E}_{k-1}& \neg\mathbf{E}_{k-1}\\
\end{bmatrix},
\end{aligned}
\end{equation}
where $\neg\mathbf{E}_{k-1}$ denotes the complement of a boolean vector $\mathbf{E}_{k-1}$, and the size of $\mathbf{E}_{k}$ is $2^{k}$. From \eqref{eq:recur}, $\mathbf{G}''_k$ can also be presented as
\begin{equation}\label{eq:G''_k2}
\mathbf{G}''_k=\begin{bmatrix}
\mathbf{G}_k\\
\mathbf{E}_k\\
\end{bmatrix}=
\begin{bmatrix}
\mathbf{0}_{2^{k-1}}& \mathbf{1}_{2^{k-1}}\\
\mathbf{G}_{k-1}&\mathbf{G}_{k-1}\\
\mathbf{E}_{k-1}& \neg\mathbf{E}_{k-1}
\end{bmatrix}.
\end{equation}

\subsection{Hamming codes}\label{sec:hamc}
The parity-check matrix $\mathbf{H}_k$ of $[2^k-1, 2^k-k-1]$ Hamming codes is defined as the concatenation of all non-zero columns in ascending order. Precisely,
\begin{equation}\label{equ:2}
\mathbf{H}_k:=\begin{bmatrix}
T^{k}(1)&T^{k}(2)&\dots&T^{k}(2^k-1)
\end{bmatrix}.
\end{equation}
The standard parity check matrix, with the form in \eqref{eq:33}, is obtained by reordering the columns of $\mathbf{H}_k$ such that the last $k$ columns form a identity matrix. Then from \eqref{eq:G0}, the generator matrix can be obtained from the standard parity check matrix accordingly. 
\if
For example, the $[7, 4]$ Hamming code has the generator matrix
\[
\bar{\mathbf{G}}_3=\begin{bmatrix}
1&0&0&0&0&1&1\\
0&1&0&0&1&0&1\\
0&0&1&0&1&1&0\\
0&0&0&1&1&1&1\\
\end{bmatrix},
\]
and the standard parity-check matrix is 
\[
\bar{\mathbf{H}}_3=\begin{bmatrix}
0&1&1&1&1&0&0\\
1&0&1&1&0&1&0\\
1&1&0&1&0&0&1\\
\end{bmatrix}.
\]
\fi

Further, the $[2^k, 2^k-k-1]$ extend Hamming code is obtained by appending an overall  parity-check to the $[2^k-1, 2^k-k-1]$ Hamming code. The parity-check matrix is defined as
\begin{equation}\label{eq:HH}
\mathbf{H}'_k:=
\begin{bmatrix}
\mathbf{0}_{k}^T&  \mathbf{H}_k\\
1&\mathbf{1}_{2^k-1}\\
\end{bmatrix}=\mathbf{G}'_k.
\end{equation}
Notably, the dual code of extended Hamming codes is the punctured Hadamard code. 
\if
For example, the parity-check matrix of $[8, 4]$ extended Hamming code is
\begin{equation}\label{equ:4}
\mathbf{H}'_3=\begin{bmatrix}
0&0&0&0&1&1&1&1\\
0&0&1&1&0&0&1&1\\
0&1&0&1&0&1&0&1\\
1&1&1&1&1&1&1&1\\
\end{bmatrix}.
\end{equation}
\fi
However, \eqref{eq:HH} is not systematic, and the systematic version is obtained by 
\[
\mathbf{H}''_k:=\mathbf{V}_{k+1} \mathbf{H}'_k=\mathbf{G}''_k,
\]
where $\mathbf{V}_{k+1}$ is defined in \eqref{eq:Vk+1}. 
\if
$\mathbf{H}''_k$ is a parity-check matrix for the Hsiao code~\cite{5391627}. 
\fi
\if
For example, $\mathbf{H}''_3$ is expressed as
\begin{equation}\label{eq:H''_3}
\mathbf{H}''_3=\mathbf{V}_4\cdot \mathbf{H}'_3=\begin{bmatrix}
0&0&0&0&1&1&1&1\\
0&0&1&1&0&0&1&1\\
0&1&0&1&0&1&0&1\\
1&0&0&1&0&1&1&0\\
\end{bmatrix},
\end{equation}
where $\mathbf{V}_4$ is shown in \eqref{eq:V4}. 
\fi
Notably, $\mathbf{H}''_k$ can be expressed as
\begin{equation}\label{eq:reH}
\mathbf{H}''_k=\begin{bmatrix}
\begin{bmatrix}
\mathbf{0}_k^T & \mathbf{H}_k
\end{bmatrix}\\
\mathbf{E}_k\\
\end{bmatrix}.
\end{equation}
\if
Further, the generator matrix can be obtained from $\mathbf{H}''_k$. For example, the reordered version of \eqref{eq:H''_3} is 
\[
\begin{bmatrix}
0&1&1&1&1&0&0&0\\
1&0&1&1&0&1&0&0\\
1&1&0&1&0&0&1&0\\
1&1&1&0&0&0&0&1\\
\end{bmatrix},
\]
where the last $4$ columns form a identity matrix. Then the generator matrix is
\[
\begin{bmatrix}
1&0&0&0&0&1&1&1\\
0&1&0&0&1&0&1&1\\
0&0&1&0&1&1&0&1\\
0&0&0&1&1&1&1&0\\
\end{bmatrix}.
\]
\fi

The encoding complexities of (extended)~Hamming codes are discussed as follows. For the $[2^k-1, 2^k-k-1]$ Hamming code, the Hamming weight of each row of $\mathbf{H}_k$ is $2^{k-1}$, and hence the number of ones in the submatrix $\mathbf{A}$~(presented in \eqref{eq:G0}) is $2^{k-1}k-k$. Thus, the number of XORs to encode Hamming codes is not more than $2^{k-1}k-k-(\text{the number of rows in }\mathbf{A})=k(2^{k-1}-2)$, which means the encoding circuit size of Hamming codes has a trivial upper bound $\mathcal{O}(n\log{n})$, where $n$ is the code length. Further, the number of XORs per message bit is not more than $k$.

Similarly, for the $[2^k, 2^k-k-1]$ extended Hamming code, as the number of ones in $\mathbf{E}_k$ in \eqref{eq:Ek} is $2^{k-1}$, the number of XORs to encoded extended Hamming codes is not more than $(k+1)(2^{k-1}-2)$. Further, the number of XORs per message bit is not more than $k+1$.

\begin{table}
	\begin{center}
		\caption{\label{tab:2}Parities associated with their checked positions of $\mathcal{C}$}
		\begin{tabular}{l|c}
			\hline\hline
			Parities & Checked position\\
			\hline
			$p_{0}$ & $0,1,3,5,7,9,11,\dots,29,31$\\
			\hline
			$p_{1}$ & $0,2 \ to \ 3,6 \ to \ 7,10 \ to \ 11,\dots,30 \ to \ 31$\\
			\hline
			$p_{2}$ & $0,4 \ to \ 7,12 \ to \ 15,20 \ to \ 23$\\
			\hline
			$p_{3}$ & $0,8 \ to \ 15,24 \ to \ 31$\\
			\hline
			$p_{4}$ & $0,16 \ to \ 31$\\
			\hline
			$p_{5}$ & $1 \ to \ 31$\\
			\hline\hline
		\end{tabular}
	\end{center}
\end{table}

\subsection{Shortened Hamming codes}\label{sec:2E}
The parity-check matrix $\mathbf{H}'=[\mathbf{A}' \quad \mathbf{I}_{k}]$ of a shortened Hamming codes $\mathcal{C}'$ with $k$ parity bits is a submatrix of a parity-check matrix $\mathbf{H}=[\mathbf{A} \quad \mathbf{I}_{k}]$ of $[2^k-1, 2^k-k-1]$ Hamming codes, where $\mathbf{A}'$ is a submatrix of $\mathbf{A}$. In shortened Hamming codes, let the message length $m$ and the number of parity bits $k$ satisfy $2^{k-1}-k<m \leq 2^{k}-k-1$. 
In Section 15.3 of \cite{Warren:2002:HD:515297}, the author presents an efficient software implementation for a shortened Hamming code with $6$ parity bits $\mathbf{p}=[p_{0} \ p_{1} \ \dots \ p_{5}]$ and $32$ message bits $\mathbf{x}=[x_{0} \ x_{1} \ \dots \ x_{31}]$. The $6$ parity bits are calculated by XORing a portion of the message bits indicated in Table \ref{tab:2}. Notably, the shortened Hamming code is then extended by appending one more parity bit $p_{6}=x_{0}+\dots+x_{31}+p_{0}+\dots+p_{5}$ to the codeword. Section \ref{sec:5D} presents the generalization of \cite{Warren:2002:HD:515297}, which is the extended version, with $k$ parity bits and $m=2^{k-1}$ message bits, of the shortened Hamming code. Furthermore, the complexity in terms of circuit size and parallel depth, of the generalized encoding algorithm is analyzed.

\section{Exact lower bounds on encoding circuit size}\label{sec:5}
In this section, we demonstrate the exact lower bounds on encoding circuit size of Hadamard codes and Hamming codes. The first two subsections show the exact lower bounds on encoding circuit size of (punctured)~Hadamard encoding algorithms, followed by the exact lower bounds on encoding circuit size of (extended)~Hamming encoding algorithms. In the rest of this paper, the circuit we considered only contains XOR operations, thus the circuit size is the number of XORs in the circuit. 

\subsection{Hadamard codes}\label{sec:3A}
\begin{theorem}\label{le:0}
A $[2^k-1, k]$ Hadamard code requires at least $2^k-k-1$ XORs in encoding process based on the generator matrix. 
\end{theorem}
\begin{proof}
The $[2^k-1, k]$ Hadamard code consists of $k$ message symbols $\{x_i\}_{i=1}^{k}$ and $2^{k}-k-1$ parities. 
Since any two columns of the generator matrix  are different, the parities of Hadamard codes are generated by different formulas, where each formula takes at least one XOR. Hence, the number of XORs taken is at least  $2^k-k-1$. 
\end{proof}

\subsection{Punctured Hadamard codes}\label{sec:3B}
This subsection discusses the encoding circuit size of punctured Hadamard codes, on the systematic version and the non-systematic version.

\subsubsection{Systematic punctured Hadamard codes}
Section~\ref{sec:hadamard} shows that the generator matrix for $[2^k, k+1]$ systematic punctured Hadamard codes consists of all  columns of size $k+1$ with odd Hamming weights. 

For $[n,k]$ systematic binary linear codes, a encoding circuit $\mathbf{x}\mathbf{G}=\begin{bmatrix}\mathbf{x}& \mathbf{x}\mathbf{A}^T\end{bmatrix}$ can be represented as a signal-flow graph $G=(V,A)$, where $V$ denotes the set of nodes, and $A$ denotes the set of arrows. Particularly, $V$ can be divided into three subsets, termed as a set of input nodes $\mathbf{o}_{in}$, output nodes $\mathbf{o}_{out}$ and hidden nodes $\mathbf{o}_{hid}$, respectively, where $\mathbf{o}_{in}$ corresponds to the message symbols, $\mathbf{o}_{out}$ corresponds to the parities, and $\mathbf{o}_{hid}$ corresponds to the intermediate results. Thus, $|\mathbf{o}_{in}|=k$, and $|\mathbf{o}_{out}|=n-k$. In the graph, each XOR operating on two nodes generates a node of $\mathbf{o}_{out}$ or $\mathbf{o}_{hid}$, 
and thus the indegree of the node of $\mathbf{o}_{out}$ or $\mathbf{o}_{hid}$ is two. In contrast, the indegree of each  $\mathbf{o}_{in}$ is zero. Figure \ref{fig:fig4} gives an example for the encoding of $[4, 3]$ punctured Hadamard codes, where $\mathbf{o}_{in}=\{x_0, x_1, x_2\}$, $\mathbf{o}_{hid}=\{z_0=x_0\oplus x_1\}$, and $\mathbf{o}_{out}=\{y_0=z_0\oplus x_2\}$. 

For a class of $[n,k]$ codes, the corresponding signal-flow graph is not unique. Let $\mathcal{G}_k$ denote a set of signal-flow graphs for the systematic $[2^k, k+1]$ punctured Hadamard codes. For a signal-flow graph $G\in \mathcal{G}_k$, $G$ depicts an algorithm with the number of XORs, $|\mathbf{o}_{out}|+|\mathbf{o}_{hid}|$. As $|\mathbf{o}_{out}|=2^k-k-1$ is a constant in $\mathcal{G}_k$, the objective is to find out a $G\in \mathcal{G}_k$, such that $|\mathbf{o}_{hid}|$ of $G$ is minimal.

From Figure \ref{fig:fig4}, it is clear that $|\mathbf{o}_{hid}|\geq 1$ in $\mathcal{G}_2$, which corresponds to $[4, 3]$ punctured Hadamard codes. This gives the following Lemma.
\begin{lemma}\label{l:2}
For any signal-flow graph of systematic punctured Hadamard codes, $\mathbf{o}_{hid}$ contains a node expressing the XOR operating on two nodes in $\mathbf{o}_{in}$.
\end{lemma}
\begin{proof}
It is obviously that the signal-flow graph has at least one node representing the XOR operating on two nodes in $\mathbf{o}_{in}$. This node is not in $\mathbf{o}_{out}$, as the Hamming weight of the encoding vector is $2$, which is not odd. Thus, this node is in $\mathbf{o}_{hid}$.
\end{proof}

For a signal-flow graph of $[n,k]$ binary linear codes, the corresponding operation on each node can be written as $\mathbf{x}a^T$, where $\mathbf{x}$ is the message vector and $a$ is the binary encoding vector. For example, in Figure \ref{fig:fig4}, we have $\mathbf{x}=[x_0\; x_1\; x_2]$, and the encoding vectors of $x_0$, $x_1$ and $x_2$ are $[1\;0\;0]$, $[0\;1\;0]$ and $[0\;0\;1]$, respectively. The encoding vector of $z_0$ is $[1\;1\;0]$, and the encoding vector of $y_0$ is $[1\;1\;1]$. The following lemma indicates that, we can merge the two nodes when both nodes have the same encoding vector.
\begin{lemma}\label{l:1.5}
For a signal-flow graph $G=(V,A)$ of $[n,k]$ binary linear codes, if there exists two nodes possessing the same encoding vector, one can obtain a new signal-flow graph $G'=(V',A')$, and $|V'|=|V|-1$.
\end{lemma}
\begin{proof}
The two nodes possessing the same encoding vector are denoted as $v_i$ and $v_j$. As $G$ has no directed cycles, the topological sort can be applied on $G$. W.L.O.G., assume $v_i$ precedes $v_j$ in the topologically sorted order. The following gives the definition of the new graph $G'=(V',A')$, where $V'=V\setminus \{v_j\}$. For each $(a,b)\in A$, $(a,b)\in A'$ if $a\neq v_j$ and $b\neq v_j$; or else, for each $(v_j,b)\in A$, then $(v_i,b)\in V'$. By removing $v_j$, the new graph $G'$ has $|V'|=|V|-1$ nodes.
\end{proof}

Recall that $\mathbf{G}_k''$ is the generator matrix of $[2^k, k+1]$ systematic punctured Hadamard codes. Then, we have the following lemma.
\begin{lemma}\label{lemma:2}
Given $\mathbf{x}=[x_0\dots x_{k-1}]$ and two integers $u, v$, $0\leq u<v < k$, a new vector $\mathbf{x}'=[x_0\dots x_{v-1}\; x_u \; x_v \dots x_{k-1}]$ is obtained by inserting $x_u$ to $\mathbf{x}$. Then $\mathbf{x}'\mathbf{G}_k''$ is a permutation of $\mathbf{x}[\mathbf{G}_{k-1}''|\mathbf{G}_{k-1}'']$.
\end{lemma}
\begin{proof}
From the definition of $\mathbf{x}'$, we have $x_u'=x_v'$. Then $\mathbf{x}'\mathbf{G}_k''$ can be written as ($G[i]$ denotes the row $i$ of $\mathbf{G}_k''$):
\[
\mathbf{x}'\mathbf{G}_k''=\sum_i x_i' G[i]=x_u'G[u]\oplus x_v'G[v]\oplus\sum_{i\neq u,v} x_i' G[i]=x_u'(G[u]\oplus G[v])\oplus\sum_{i\neq u,v} x_i' G[i]=\mathbf{x}\bar{\mathbf{G}}_k,
\]
where $\mathbf{x}=[x_0\dots x_{k-1}]$, and 
\begin{equation}\label{eq:barG}
\bar{\mathbf{G}}_k=
\begin{bmatrix}
G[0]\\
\vdots \\
G[u]\oplus G[v] \\
\vdots \\
G[v-1]\\
G[v+1]\\
\vdots \\
G[k]\\
\end{bmatrix}.
\end{equation}

To prove $\mathbf{x}\bar{\mathbf{G}}_k$ is a permutation of $\mathbf{x}[\mathbf{G}_{k-1}''|\mathbf{G}_{k-1}'']$, this is equivalent to show that, for each column $g^T$ of $\mathbf{G}_{k-1}''$, there exist two columns of $\bar{\mathbf{G}}_k$, and the content of both columns are the same with $g^T$.

For every column $g^T=[g_0\dots g_{k-1}]^T$ of $\mathbf{G}_{k-1}''$, two $(k+1)$-element column vectors are defined as
\begin{equation}
\begin{aligned}
\bar{g}_0^T&=[g_0\dots g_{v-1}\; 0\; g_{v}\dots g_{k-1}]^T,\\
\bar{g}_1^T&=[g_0\dots g_u\oplus 1\dots g_{v-1}\; 1\; g_{v}\dots g_{k-1}]^T.
\end{aligned}
\end{equation}
As the Hamming weight of $g$ is odd, the Hamming weights of both $\bar{g}_0$ and $\bar{g}_1$ are also odd, and hence $\bar{g}_0^T$ and $\bar{g}_1^T$ are two distinct columns of $\mathbf{G}_{k}''$. Further, \eqref{eq:barG} shows that the $\bar{g}_0^T$ and $\bar{g}_1^T$ become $g^T$ in $\bar{\mathbf{G}}_k$. This completes the proof.
\end{proof}

\begin{lemma}\label{lextra:0}
Given $G_k\in \mathcal{G}_k$ with $m$ hidden nodes, there exists a  $G_{k-1}\in \mathcal{G}_{k-1}$ with the number of hidden nodes at most $m-1$.
\end{lemma}
\begin{proof}
For any $G_k\in \mathcal{G}_k$, Lemma~\ref{l:2} shows that contains a hidden node expressing the XOR of two input nodes. Assume the two nodes correspond to $x_u$ and $x_v$ of the message vector $\mathbf{x}$. From Lemma~\ref{lemma:2}, $G_k$ can be converted to the encoding of $[2^{k-1}, k]$ systematic punctured Hadamard codes. The message vector $\mathbf{x}$ is converted to $\mathbf{x}'=[x_0\dots x_{v-1}\; x_u \; x_v \dots x_{k-1}]$, which is the input of $G_k$. Then the output is $\mathbf{x}[\mathbf{G}_{k-1}''|\mathbf{G}_{k-1}'']$ in a specific order. 

The above signal-flow graph for $[2^{k-1}, k]$ systematic punctured Hadamard codes can be simplified further. First, the number of input and output nodes in $G_k$ is $2^k$. By the merge operation in Lemma~\ref{l:1.5}, the number of input and output nodes can be reduced to $2^{k-1}$. Second, the hidden node expressing the XOR of two input nodes can be eliminated, as the node expresses $0=x_u\oplus x_u$ on $\mathbf{x}'$. As a result, in the new signal-flow graph, the number of output and input nodes is $2^{k-1}$, and the number of hidden nodes is at most $m-1$.
\end{proof}

\begin{lemma}\label{l:3}
For any $G_k\in \mathcal{G}_k$, the number of hidden nodes of $G_k$ is at least $k-1$.
\end{lemma} 
\begin{proof}
Assume there exists a $G_k'\in \mathcal{G}_k$ with $m$ hidden nodes. From Lemma~\ref{lextra:0}, one can obtain $G_{k-1}'\in \mathcal{G}_{k-1}$ with at most $m-1$ hidden nodes. In this way, one can obtain $G_2'\in \mathcal{G}_2$ with at most $m-k+2$ hidden nodes. From Lemma~\ref{l:2}, we have
\[
m-k+2\geq 1 \Rightarrow m\geq k-1.
\]
This completes the proof.
\end{proof}

With Lemma~\ref{l:3}, the lower bound is given below.
\begin{theorem}\label{le:1}
The  $[2^k, k+1]$ systematic punctured Hadamard code requires at least $2^k-2$ XORs in encoding process based on the generator matrix.
\end{theorem}
\begin{proof}
The encoding of $[2^k, k+1]$ systematic punctured Hadamard codes is expressed an a signal-flow graph $G_k\in\mathcal{G}_k$. From Lemma~\ref{l:3}, $G_k$ has $2^k-k-1$ output nodes and at least $k-1$ hidden nodes. Thus, it requires $(2^k-k-1)+(k-1)=2^k-2$ XORs in encoding. This completes the proof.
\end{proof}


\subsubsection{Non-systematic punctured Hadamard codes}
The generator matrix $\mathbf{G}_k'$ of non-systematic punctured Hadamard code~(NPHC) is defined in \eqref{eq:GG}. 

From \eqref{eq:GG}, there is only one column of $\mathbf{G}_k'$ which has Hamming weight $1$. For example, for $\mathbf{G}_3'$ in \eqref{eq:G3}, the first column is $[0\; 0\; 0\; 1]^T$. The other $2^k-1$ symbols on the codeword shall be computed. Refer to Theorem \ref{le:0}, we get a theorem below.

\begin{theorem}\label{l:4}
A $[2^k, k+1]$ non-systematic punctured Hadamard code requires at least $2^k-1$ XORs in encoding process based on the generator matrix.
\end{theorem}
\begin{proof}
The $[2^k, k+1]$ non-systematic punctured Hadamard code consists of a message symbol $x_0$ and $2^k-1$ computed symbols. Assume the encoding algorithm takes $S$ XORs, and those XORs produce $S$ symbol $\mathcal{G}=\{g_i\}_{i=1}^{S}$. Then each symbol of the codeword is chosen from $\mathcal{G}$. As any two columns of $\mathbf{G}_k'$ are different, this leads that $S\geq 2^k-1$. This completes the proof.
\end{proof}

\subsection{Hamming codes}\label{sec:CHam}
First, the encoding circuit size of Hamming codes is discussed. Then the encoding circuit size of extended Hamming codes is considered.

The proof of lower bound relies on the property of dual codes and the transposition principle, which states that if an algorithm for a matrix-vector product by a matrix $M$ exists, then there exists an algorithm for a matrix-vector product by its transpose $M^T$ in similar complexity. A formal definition is stated below.
\begin{theorem}\label{th:1}
(Transposition principle~\cite{1956-bordewijk}): Given an $i$-by-$j$ matrix $M$ without zero rows or columns, let $a(M)$ denote the minimum number of operations to compute the product $v_i M$ with a vector $v_i$ of size $i$. Then there exists an algorithm to compute $v_j M^T$ in $a(M)+j - i$ arithmetic operations.
\end{theorem}

\subsubsection{Hamming codes}
This part shows the numbers of XORs used in the encoding of $[2^k-1, 2^k-k-1]$ Hamming codes.

By Theorem \ref{th:1}, we have the main theorem as follows.
\begin{theorem}\label{th:2}
A $[2^k-1, 2^k-k-1]$ Hamming code requires at least $2^{k+1}-3k-2$ XORs in encoding process based on the generator matrix.
\end{theorem} 
\begin{proof}
The statement is proved by contradiction. Assume the encoding of Hamming codes $\mathbf{x}_1\mathbf{G}=[\mathbf{x}_1\; \mathbf{x}_1\mathbf{A}]$ only requires $2^{k+1}-3k-2-\epsilon$ XORs, where $\epsilon$ is a positive integer and $\mathbf{G}$ is the standard generator matrix of A $[2^k-1, 2^k-k-1]$ Hamming codes. From Theorem \ref{th:1}, there exists an algorithm to compute $\mathbf{x}_2\mathbf{A}^T$ in 
\[
2^{k+1}-3k-2-\epsilon+(k-(2^{k}-k-1))=2^{k}-k-1-\epsilon
\]
XORs. Note that  $[\mathbf{A}^T\ \mathbf{I}]$ is the generator matrix of the Hadamard codes and the above result contradicts Theorem \ref{le:0}. 
\end{proof}

\subsubsection{Extended Hamming codes}
Like Theorem \ref{th:2}, we present corresponded Theorem for extended Hamming code. 
\begin{theorem}\label{th:3}
A $[2^{k}, 2^{k}-k-1]$ extended Hamming code requires at least $2^{k+1}-2k-4$ XORs in encoding.
\end{theorem}
Proof of Theorem \ref{th:3} is similar to that for Theorem~\ref{th:2} such that we omit it. 

\section{Encoding Algorithms for (punctured) Hadamard codes}\label{sec:4.0}
This section presents the encoding algorithms of $[2^k, k]$ Hadamard codes and $[2^k, k+1]$ punctured Hadamard codes, and analyzes their circuit sizes and parallel depth. We also show that the proposed algorithms achieves these lower bounds of circuit sizes derived in Section \ref{sec:3A} and \ref{sec:3B}.

\subsection{Encoding algorithm with Gray-code ordering}
First, the columns of $\mathbf{G}_{k}$ are permuted by the Gray codes. For example, $\mathbf{G}_{3}$ can be rearranged as
\begin{equation}\label{eq:17}
\bar{\mathbf{G}}_{3}=\left[
\begin{array}{c|cc|cccc|c}
0& 0 & 0 & 0& 1&1& 1& 1\\
0&0&1&1&1&1&0&0\\
0&1&1&0&0&1&1&0\\
\end{array}
\right].
\end{equation}
Given $\mathbf{x}=[x_0\ x_1\  x_{2}]$, the codeword $\mathbf{y}=[y_0\ \cdots\  y_{7}]$ is obtained via $\mathbf{y}=\mathbf{x}\bar{\mathbf{G}}_{3}$. As shown in \eqref{eq:17}, there is only one different bit in each pair of adjacent columns. Based on this property, each $y_i$ can be obtained by XORing $y_{i-1}$ and a message bit. For example, $y_{2}=y_{1}+x_{1}$ and $y_{4}=y_{3}+x_{0}$. Specifically, as $y_{1}=x_{2}$, $y_{3}=x_{1}$, $y_{7}=x_{0}$ are message symbols, it does not  require XOR operations. Thus, the Hadamard codes can be encoded with $2^{k}-k-1$ XORs, which achieves the lower bound given in Theorem \ref{le:0}.

Note that $\bar{\mathbf{G}}_{k}$ can be divided into $k+1$ submatrices as shown in \eqref{eq:17}, where the first block is a zero column, and other $k$ blocks are started with the column corresponding to a message bit. Clearly, the encoding for each block can be calculated simultaneously. However, the symbols in each block shall be calculated serially, and thus the parallel depth is $p_{1}=(size \ of \ largest \ block) -1$. From the pigeon hole principle, the size of largest block is larger or equal to $(2^{k}-1)/k$, thus $p_{1} \geq (2^{k}-1)/k-1$.

$\bar{\mathbf{G}}_{k}$ can be extended into the generator matrix $\bar{\mathbf{G}}_{k}''$ for the systematic $[2^k, k+1]$ punctured Hadamard codes. For example, given $\mathbf{x}=[x_0\ \cdots\ x_{3}]$, the codeword $\mathbf{y}=[y_0\ \cdots\ y_{7}]$ is obtained via $\mathbf{y}=\mathbf{x}\bar{\mathbf{G}}_{3}''$, where
\begin{equation}\label{eq:18}
\bar{\mathbf{G}}_{3}''=\left[
\begin{array}{c|cc|cccc|c}
0& 0 & 0 & 0& 1&1& 1& 1\\
0&0&1&1&1&1&0&0\\
0&1&1&0&0&1&1&0\\
1&0&1&0&1&0&1&0\\
\end{array}
\right].
\end{equation}
As shown in \eqref{eq:18}, there are two different bits in each pair of adjacent columns such that each $y_{i}$ can be obtained by XORing $y_{i-1}$, $x_{3}$, and another message bit. For example, $y_{4}=y_{3}+x_{0}+x_{3}$. Let $x_{i}'=x_{i}+x_{3}$, for $i\in[0,k-1]$. This requires $k$ XORs. Then each $y_{i}$ can be obtained by XORing $y_{i-1}$ and a symbol of $\{x_{i}'\}_{i\in[0,k-1]}$. For example, $y_{4}=y_{3}+x_{0}'$. Thus, the $[2^k, k+1]$ punctured Hadamard codes can be encoded with $2^{k}-1$ XORs, which is larger than the lower bound showed in Theorem \ref{le:1}. Similarly, the parallel depth of this encoding algorithm is larger or equal to $(2^{k}-1)/k$.

\subsection{Proposed encoding algorithms}\label{sec:eee}
Given the message vector $\mathbf{x}_k=[x_0\dots x_{k-1}]$, the codeword $\mathbf{y}_k=[y_0\dots y_{2^{k}-1}]$ is defined as 
\begin{equation}\label{eq:defyk}
\mathbf{y}_k=\mathbf{x}_k\mathbf{G_k}.
\end{equation}
\subsubsection{Hadamard codes}
The proposed encoding algorithm is based on the recursive structure of the generator matrix \eqref{eq:recur}. We divide $\mathbf{x}_k$ and $\mathbf{y}_k$ into
\begin{equation}\label{eq:xk}
\mathbf{x}_k=[x_0|\mathbf{x}'_k],
\end{equation}
\begin{equation}\label{eq:yk}
\mathbf{y}_k=[\mathbf{y}_{k-1}|\mathbf{y}'_{k-1}],
\end{equation}
where $\mathbf{x}'_k=[x_1\dots x_{k-1}]$, $\mathbf{y}_{k-1}=[y_0\dots y_{2^{k-1}-1}]$ and $\mathbf{y}'_{k-1}=[y_{2^{k-1}}\dots y_{2^{k}-1}]$. From \eqref{eq:recur}, we have
\begin{equation}\label{eq:had20}
\mathbf{y}_{k-1}=\mathbf{x}'_k\mathbf{G_{k-1}},
\end{equation}
\begin{equation}\label{eq:had2}
\begin{aligned}
\mathbf{y}'_{k-1}&=\mathbf{x}'_k\mathbf{G_{k-1}}\oplus [x_0\dots x_0]=\mathbf{y}_{k-1}\oplus [x_0\dots x_0],
\end{aligned}
\end{equation}
where $\oplus$ denotes the bitwise XOR.
\eqref{eq:had2} shows that $\mathbf{y}'_{k-1}$ can be computed from $\mathbf{y}_{k-1}$ in $1$ XOR per bit. Based on~\eqref{eq:defyk} and \eqref{eq:had20}, like computing $\mathbf{y}_k$, we can apply the same approach on $\mathbf{y}_{k-1}$. Thus, $\mathbf{y}_k$ can be computed recursively with the basis case $\mathbf{y}_{1}=x_{k-1}[0 \quad 1]$. The algorithm is presented in Algorithm~\ref{alg:5}. Note that, from $\mathbf{G_{k-1}}$, it can be seen that the first element in $\mathbf{y}'_{k-1}$ is $x_0$ and this is true for each iteration in the recursion. Hence, in line 5 of Algorithm~\ref{alg:5}, we can assign $x_0$ to the first element of $\mathbf{y}_{2}$ without performing the XOR operation.

\begin{algorithm}[t]
\caption{$\mathrm{P}_{1}(\mathbf{x}, k)$: Encoding of Hadamard codes}\label{alg:5}
\KwIn{$\mathbf{x}=[x_0\quad x_1\quad  \dots \quad  x_{k-1}]$, and the message length $k$}
\KwOut{$\mathbf{y}=[y_0\quad y_1\quad  \dots \quad  y_{2^{k}-1}]$}
\If{$k=1$}{
\Return $[0\quad x_0]$
}
Call $\mathbf{y}_1\leftarrow \mathrm{P}_{1}(\mathbf{x}', k-1)$, where $\mathbf{x}'=[x_1  \dots  x_{k-1}]$\\
$\mathbf{y}_2\leftarrow 
\begin{bmatrix}
x_0 & \mathbf{y}_1[1]\oplus x_0 &\dots &\mathbf{y}_1[2^{k-1}-1]\oplus x_0
\end{bmatrix}$\\
\Return $[\mathbf{y}_1 \quad\mathbf{y}_2]$\\
\end{algorithm}

\subsubsection{Punctured Hadamard codes}\label{sec:PHC}
For the non-systematic version \eqref{eq:GG}, the algorithm is similar to Algorithm~\ref{alg:5}. Given the message vector $\mathbf{x}=[x_0\quad x_1\quad  \dots \quad  x_{k}]$ of size $k+1$, the recursive operation of the algorithm will stop when $k=1$. Line 2 in Algorithm \ref{alg:5} is then modified as $\mathbf{y}\leftarrow [x_1\quad x_1\oplus x_0]$ and the $2^{k}$-element output $\mathbf{y}$ becomes the codeword. 

For simplicity, with a little abuse of notations, we use the same notations as the previous section to denote the message and codeword in this section. For the systematic version given in~\eqref{eq:G''_k2}, the codeword $\mathbf{y}_k$ is defined as 
\begin{equation}\label{eq:y_k}
\mathbf{y}_k=\mathbf{x}_k \mathbf{G}''_{k}=[y_0\quad \dots y_{2^{k}-1}],
\end{equation}
where $\mathbf{x}_{k}=[x_0 \dots  x_{k}]$ is the message vector. 
\if
Similar to \eqref{eq:xk} and \eqref{eq:yk}, $\mathbf{x}_k$ is divided into $x_0$ and $\mathbf{x}'_k=[x_1 \dots x_k]$, and $\mathbf{y}_k$ is divided into two equal parts $\mathbf{y}_{k-1}$ and $\mathbf{y}'_{k-1}$. 
\fi
From \eqref{eq:G''_k2}, \eqref{eq:xk} and \eqref{eq:yk}, \eqref{eq:y_k} can be written as
\begin{equation}
\begin{aligned}
&\left[ \mathbf{y}_{k-1}\; \mathbf{y}'_{k-1} \right]
=\left[ x_0\; \mathbf{x}'_k \right] 
\begin{bmatrix}
\mathbf{0}_{2^{k-1}}& \mathbf{1}_{2^{k-1}}\\
\mathbf{G}_{k-1}&\mathbf{G}_{k-1}\\
\mathbf{E}_{k-1}& \neg\mathbf{E}_{k-1}
\end{bmatrix}
\end{aligned}
\end{equation}
This gives
\begin{equation}\label{eq:recPHad}
\begin{aligned}
&\mathbf{y}_{k-1}
=\mathbf{x}'_k
\begin{bmatrix}
\mathbf{G}_{k-1}\\
\mathbf{E}_{k-1}
\end{bmatrix}
=\mathbf{x}'_k \mathbf{G}''_{k-1},\\
&\mathbf{y}'_{k-1}=[x_0\dots x_0]\oplus \mathbf{x}'_k
\begin{bmatrix}
\mathbf{G}_{k-1}\\
\neg\mathbf{E}_{k-1}
\end{bmatrix}=
\mathbf{y}_{k-1} \oplus [x_0\dots x_0] \oplus [x_k\dots x_k],
\end{aligned}
\end{equation}
In \eqref{eq:recPHad}, $\mathbf{y}_{k-1}$ can be computed with the same approach, recursively. The basis case is $\mathbf{y}_1=[x_k, x_{k-1}]$. The proposed encoding algorithm is given in Algorithm~\ref{alg:6}. Similar to Algorithm~\ref{alg:5}, from $\mathbf{G}''_{k-1}$, it can be seen that the first element in $\mathbf{y}'_{k-1}$ is $x_0$ and this is true for each iteration in the recursion. Hence, in line 6 of Algorithm~\ref{alg:6}, we can assign $x_0$ to the first element of $\mathbf{y}_{2}$ without performing the XOR operation.

\begin{algorithm}[t]
\caption{$\mathrm{P}_{2}(\mathbf{x}, k)$: Encoding of punctured Hadamard codes}\label{alg:6}
\KwIn{$\mathbf{x}=[x_0\quad x_1\quad  \dots \quad  x_{k}]$, and $k$}
\KwOut{$\mathbf{y}=[y_0\quad y_1\quad  \dots \quad  y_{2^{k}-1}]$}
\If{$k=1$}{
\Return $[x_1 \quad x_0]$
}
Call $\mathbf{y}_1\leftarrow \mathrm{P}_{2}(\mathbf{x}', k-1)$, where $\mathbf{x}'=[x_1  \dots  x_{k}]$\\
$t \leftarrow x_0 \oplus x_k$\\
$\mathbf{y}_2\leftarrow 
\begin{bmatrix}
x_0 & \mathbf{y}_1[1]\oplus t &\dots &\mathbf{y}_1[2^{k-1}-1]\oplus t
\end{bmatrix}$\\

\Return $[\mathbf{y}_1 \quad\mathbf{y}_2]$\\
\end{algorithm}

\subsection{Circuit sizes and the parallel depth}
\subsubsection{Hadamard codes}
The the number of XORs (circuit size) in Algorithm~\ref{alg:5} for $[2^{k}, k]$ Hadamard codes is denoted as $A_1(k)$. In Algorithm~\ref{alg:5}, Lines $1-3$ give the base case
\begin{equation}
A_1(1)=0.
\end{equation}
Line $4$ calls the procedure recursively. Note that since the first element in $\mathbf{y}_1$ is always $0$, the first element in $\mathbf{y}_2$ is always $x_0$. Hence, line $5$ requires $2^{k-1}-1$ additions as the first XOR among them can be replaced by the assignment. In summary, the recurrence relation is written as
\begin{equation}
A_1(k)=A_1(k-1)+2^{k-1}-1,
\end{equation}
and the solution is
\begin{equation}
\label{complexity-H}
\begin{aligned}
A_1(k)=2^{k}-k-1.
\end{aligned}
\end{equation}
For the parallel depth of Algorithm \ref{alg:5}, all the calculations in line $5$ can be paralleled. Thus, the parallel depth of Algorithm \ref{alg:5} is $k-1$, its recursive depth. 

\subsubsection{Punctured Hadamard codes}
The number of XORs for systematic codes in Algorithm~\ref{alg:6} is denoted as $A_2(k)$. In Algorithm~\ref{alg:6}, Lines $1-3$ give the base case
\begin{equation}
A_2(1)=0.
\end{equation}
Line $4$ calls the procedure recursively. According to the analysis of $A_1(k)$, line $5$ and $6$ require $2^{k-1}$ additions. Thus,
\begin{equation}
A_2(k)=A_2(k-1)+2^{k-1},
\end{equation}
and the solution is
\begin{equation}
\label{complexity-NSPH}
\begin{aligned}
A_2(k)=2^{k}-2.
\end{aligned}
\end{equation}
Now we discuss the parallel depth of Algorithm \ref{alg:6}. Note that line $5$ can be calculated in advance, which requires an extra layer in circuit. In addition to this extra layer,  Algorithm \ref{alg:5} and Algorithm \ref{alg:6} are  the same. since Algorithm \ref{alg:5}' parallel depth is $k-1$, the parallel depth of Algorithm \ref{alg:6} is $k$.

The complexities for the non-systemetic codes $A'_2(k)$ in Section \ref{sec:eee} is similar to Algorithm \ref{alg:5}. The base case requires $1$ XOR. 
\begin{equation}
A'_2(1)=1,
\end{equation}
\begin{equation}
A'_2(k)=A'_2(k-1)+2^{k-1},
\end{equation}
and hence
\begin{equation}
\label{complexity-SPH}
A'_2(k)=2^{k}-1.
\end{equation}
Thus, it requires $2^{k}-2$ XORs and $2^{k}-1$ XORs in the encoding of systematic and non-systematic punctured Hadamard codes, respectively.

In conclusion, from \eqref{complexity-H}, \eqref{complexity-NSPH}, and \eqref{complexity-SPH}, the lower bounds on encoding circuit size of Hadamard codes, presented in Theorem \ref{le:0}, \ref{le:1}, and \ref{l:4}, are achievable.

\section{Novel Encoding Algorithms for (Extended) Hamming codes}\label{sec:3}
This section presents the encoding algorithms of Hamming codes and extended Hamming code, where the recursive versions of the proposed algorithms are presented. We also show that the proposed algorithms achieves these lower bounds derived in Section \ref{sec:CHam}.


\subsection{Hamming Codes}
Given a codeword $\mathbf{y}\in \mathbb{F}_2^n$ of $[2^k-1, 2^k-k-1]$ Hamming codes, we have 
\begin{equation}\label{eq:43}
\mathbf{H}_k \mathbf{y}=0,
\end{equation}
where $\mathbf{H}_k$ is defined in \eqref{equ:2}. 
Then the encoding process can be described as follows. First, given $2^{k}-k-1$ message bits, we construct a $2^{k}$-element column vector $\mathbf{x}=[x_{0} \quad x_{1} \dots \quad x_{2^k-1}]^T$, where $x_{0}=0$, $x_{2^i}=0$, for $i=0,1,\dots ,k-1$, and other $2^{k}-k-1$ bits are message bits.  The bottom layer in Figure~\ref{fig:fig1} is an example of $\mathbf{x}$, with $k=3$, where the white nodes are set to be zero and the black nodes are filled with message bits. Then the parities are calculated by
\begin{equation}\label{equ:7_}
\mathbf{p}=[p_{1} \dots p_{k}]^T=\begin{bmatrix}
\mathbf{0}_k^{T} &\mathbf{H}_k
\end{bmatrix}\mathbf{x}.
\end{equation}
Then, let $x_{2^{i}}=p_{k-i}$ with $i=0,1,\dots,k-1$ and it can be verified that $\mathbf{H}_{k}[x_{1}\dots x_{n}]^{T}=0$, i.e., $[x_{1}\dots x_{n}]^{T}$ is the corresponding codeword. Further, we discuss how to calculate $\mathbf{p}$ more effectively. From \eqref{eq:Tki} and \eqref{equ:2}, \eqref{equ:7_} can be written as 
\begin{equation}\label{equ:12}
\begin{aligned}
\mathbf{p}&=\sum_{i=0}^{2^{k-1}-1}{T^{k}(i)x_{i}}+\sum_{i=0}^{2^{k-1}-1}{(T^{k}(i)+[1\ 0 \dots 0]^T)x_{i+2^{k-1}}}\\
&=\sum_{i=0}^{2^{k-1}-1}{T^{k}(i)(x_{i}+x_{i+2^{k-1}})}+[1\ 0 \dots 0]^T\sum_{i=2^{k-1}}^{2^{k}-1}{x_{i}}.
\end{aligned}
\end{equation}
As the first elements in $T^k(i)$ is $0$, with $0\le i\le 2^{k-1}-1$,  the first element in the column vector $\sum_{i=0}^{2^{k-1}-1}{T^{k}(i)(x_{i}+x_{i+2^{k-1}})}$ is $0$. Thus, from~\eqref{equ:7_} and~\eqref{equ:12}, we have 
\begin{equation}\label{eq:27}
p_1=x_{2^{k-1}}+\dots+x_{2^k-1}
\end{equation}
and
\begin{equation}\label{eq:28}
\begin{aligned}
\mathbf{p}'=[p_{2} \dots p_{k}]^{T}=\sum_{i=0}^{2^{k-1}-1}{T^{k}(i)(x_{i}+x_{i+2^{k-1}})}.
\end{aligned}
\end{equation}
Let $\mathbf{x}'=[x_0'\dots x_{2^{k-1}-1}']^T$ denote a $2^{k-1}$-element vector, where
\begin{equation}\label{eq:16}
x_{i}'=x_i+x_{2^{k-1}+i} \qquad\forall i\in [0, 2^{k-1}-1].
\end{equation}
Then \eqref{eq:28} can be written as
\begin{equation}\label{eq:29}
\begin{aligned}
\mathbf{p}'=&\sum_{i=0}^{2^{k-1}-1}{T^{k-1}(i)x_{i}'}=\begin{bmatrix}
\mathbf{0}_{k-1}^{T} &\mathbf{H}_{k-1}
\end{bmatrix}\mathbf{x}',
\end{aligned}
\end{equation}
which can be computed recursively by applying the same approach on $\mathbf{p}$.

\if
Figure \ref{fig:fig0} presents the signal-flow graph of calculating $\mathbf{H}_3\mathbf{x}^0$. 
\begin{figure}
\center
   \includegraphics[width=0.7\columnwidth]{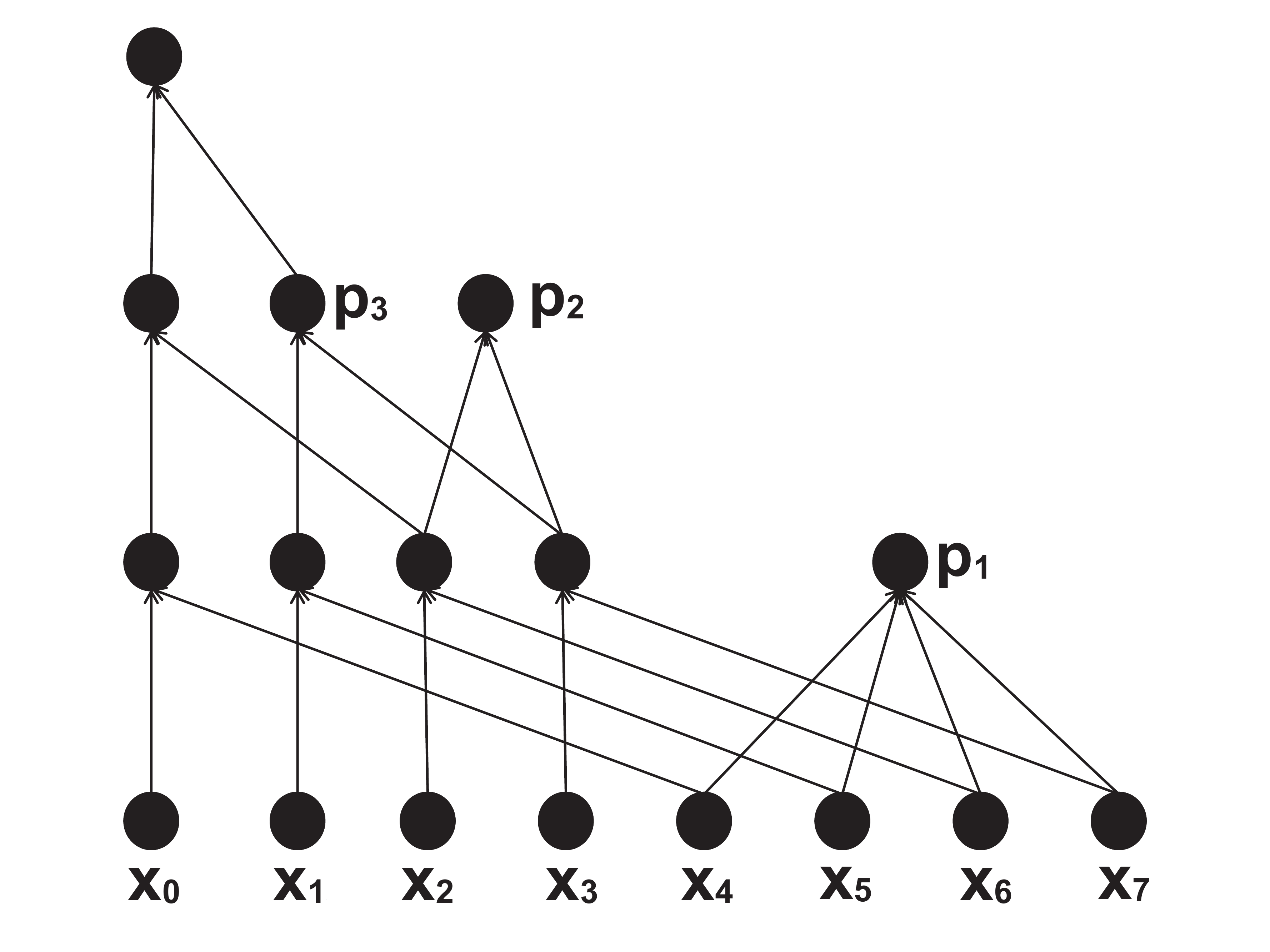}
\caption{\label{fig:fig0} Signal-flow graph of calculating $\mathbf{H}_3\mathbf{x}^0$}
\end{figure}
\fi


\begin{algorithm}[t]
\caption{$\mathrm{P}_{3}(\mathbf{x}, k)$: Encoding of Hamming codes (Recursive version)}\label{alg:1}
\KwIn{$\mathbf{x}=[x_0\quad  x_1\quad \dots \quad x_{2^{k}-1}]$, and $k$}
\KwOut{$\mathbf{p}$}
\If{$k=1$}{
\Return $x_1$
}
$S \leftarrow \sum_{i=0}^{2^{k-1}-1}x_{2^{k-1}+i}$\\
\For{$i= 0$ \KwTo $(2^{k-1}-1)$}{
$x_i'\leftarrow x_i+ x_{2^{k-1}+i}$
}
Call $\mathbf{p}' \leftarrow \mathrm{P}_{3}(\mathbf{x}', k-1)$,
where $\mathbf{x}'=[x_0'\quad x_1'\quad \dots\quad x_{2^{k-1}-1}']$\\
\Return $[S \quad \mathbf{p}']$
\end{algorithm}

Algorithm~\ref{alg:1} depicts the explicit steps, where the input $\mathbf{x}$ shall satisfy the constraints $x_0=0$ and $x_{2^i}=0,\forall i\in [0, k-1]$. In the algorithm, Lines 1-3 handle the basis case $k=1$. Line 4 calculates \eqref{eq:27}. Lines 5-7 perform \eqref{eq:16}. Line 8 calls the procedure recursively on the new input $\mathbf{x}'$. Finally, Line 9 returns the result, which is the vector of parities. Figure~\ref{fig:fig1} presents the circuit of the proposed algorithm at $k=3$, where the white nodes store zeros $x_0=x_1=x_2=x_4=0$, and other nodes are filled with message bits.

\begin{figure}
\center
   \includegraphics[width=0.7\columnwidth]{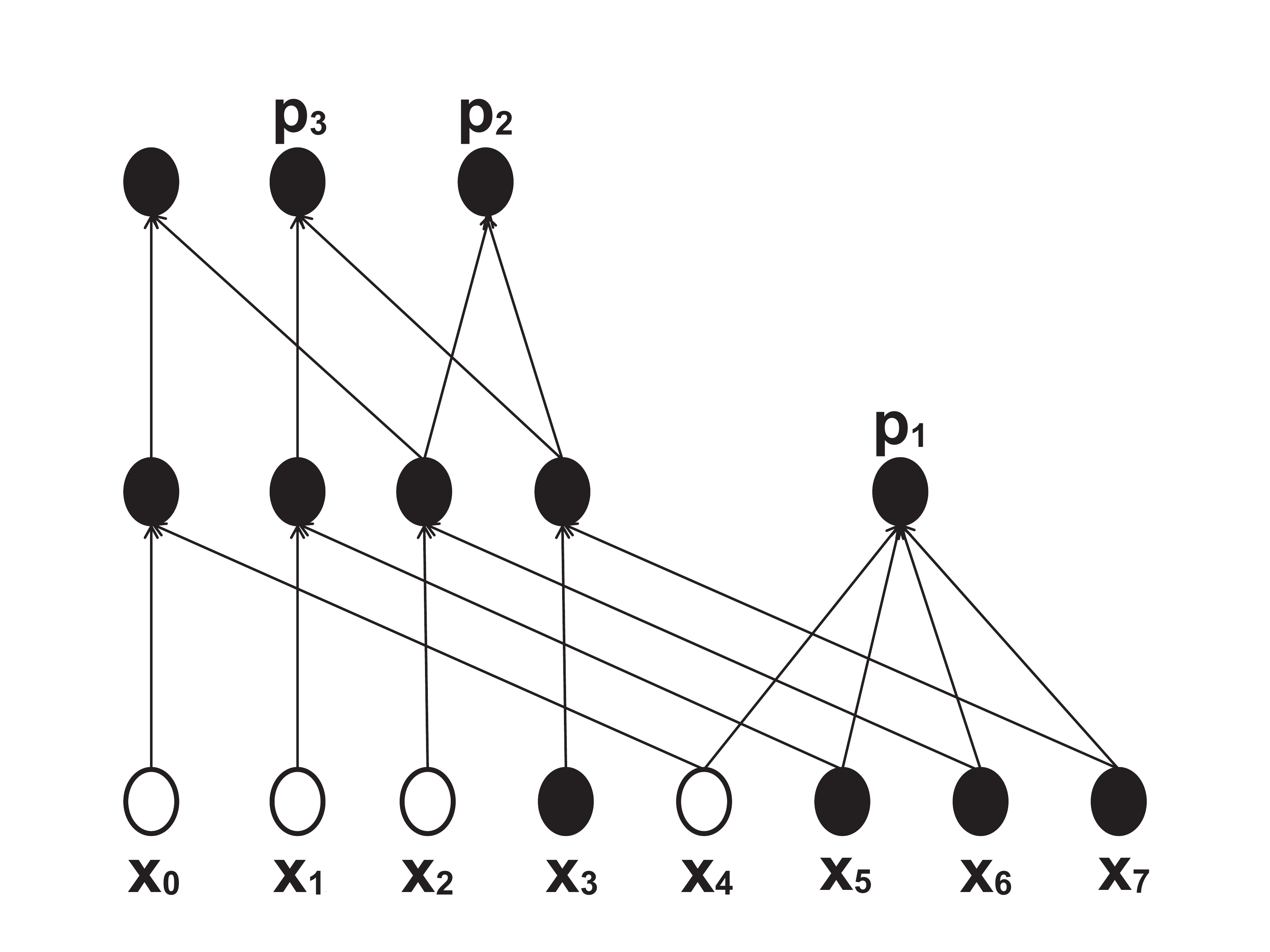}
\caption{\label{fig:fig1} Circuit of Algorithm \ref{alg:1} with $k=3$}
\end{figure}

\subsection{Extended Hamming Codes}
Similarly with Hamming codes, the encoding process of $[2^k, 2^k-k-1]$ extended Hamming codes can be described as follows. First, we construct the $\mathbf{x}$ from $2^{k}-k-1$ message bits and calculate $\mathbf{p}$ from~\eqref{equ:7_}, similarly. Then, we calculate the extra parity bit $p_{k+1}$ via
\begin{equation}\label{eq:PI}
p_{k+1}=\sum_{i=0}^{2^k-1}{x_i}+\sum_{i=1}^{k}{p_i}.
\end{equation}
Then, let $x_{0}=p_{k+1}$ and $x_{2^{i}}=p_{k-i}$ with $i=0,1,\dots,k-1$. Thus from~\eqref{eq:HH} and~\eqref{eq:PI}, it can be verified that
\begin{equation}
\mathbf{H}_k' \mathbf{x}=\begin{bmatrix}
\mathbf{0}_{k}^T &  \mathbf{H}_k\\
1&\mathbf{1}_{2^k-1}\\
\end{bmatrix} \mathbf{x}=0,
\end{equation}
which means $\mathbf{x}$ is the corresponding codeword.

Next we present an approach that incorporates the calculation of $p_{k+1}$ into the recursive formula to reduce the complexity. Let $\mathbf{x}^{i}=[x_{0}^{i}\dots x_{2^{k-i}-1}^{i}]^T, \forall i \in [0, k]$, denote a $2^{k-i}$-element vector and 
\begin{equation}\label{eq:TII}
t_{i}= x_{2^{k-i-1}+1}^{i}+\dots+x_{2^{k-i}}^{i},
\end{equation}
which is the summation of the last $2^{k-i-1}-1$ elements of $\mathbf{x}^{i}$. Let $\mathbf{x}^{0}=\mathbf{x}$. Accordingly, $\mathbf{x}_{i}$ can be obtained from $\mathbf{x}_{i-1}$ via
\begin{equation}\label{eq:49}
x_{j}^{i}=
\begin{cases}
x_{0}^{i-1}+t_{i-1} & j = 0;\\
x_j^{i-1}+x_{2^{k-i}+j}^{i-1} & \forall j\in [2^{k-i}-1].
\end{cases}
\end{equation}
Then, we can obtain parities via
\begin{equation}\label{eq:pi}
p_{i}=
\begin{cases}
x_{2^{k-i}}^{i-1}+t_{i-1} & i \in [0,k-1];\\
x_{1}^{k-1} & i=k;\\
x_{0}^{k-1} & i=k+1.\\
\end{cases}
\end{equation}
For example, Figure~\ref{fig:fig2} presents the information flow graph of the calculation of all $k+1$ parities with $k=3$.

Then, we will show that $p_{1},\dots,p_{k}$ satisfies~\eqref{equ:7_} and $p_{k+1}$ satisfies~\eqref{eq:PI}. First, from~\eqref{eq:TII} and~\eqref{eq:pi}, we have $p_{i}=\sum_{i=2^{k-i}}^{2^{k-i+1}-1}{x_{i}^{i-1}}$, with $i \in [k]$. Thus, $p_{i}$, with $i \in [k]$, is the summation of the last half bits of $\mathbf{x}^{i-1}$, which is the same as the recursive calculation of~\eqref{equ:7_}. 

Second, from~\eqref{eq:pi}, we have 
\begin{equation}\label{eq:pk1}
p_{k+1}=p_{k+1}+p_{k}+p_{k}=x_{0}^{k-1}+p_{k}+x_{1}^{k-2}+x_{3}^{k-2}.
\end{equation}
Then, from~\eqref{eq:49} and~\eqref{eq:pi}, we have
\begin{equation}\label{eq:x0i}
x_{0}^{i}=x_{0}^{i-1}+t_{i-1}=x_{0}^{i-1}+p_{i}+x_{k-i}^{i-1}.
\end{equation}
From~\eqref{eq:x0i},~\eqref{eq:pk1} becomes
\begin{equation}\label{eq:pk1'}
\begin{aligned}
p_{k+1}&=x_{0}^{k-2}+x_{2}^{k-2}+p_{k-1}+x_{1}^{k-2}+x_{3}^{k-2}+p_{k}.
\end{aligned}
\end{equation} 
Further, from~\eqref{eq:49} and~\eqref{eq:x0i}, we have
\begin{equation}\label{eq:sx}
\sum_{j\in [0,2^{k-i}-1]}{x_{j}^{i}}=\sum_{j\in [0,2^{k-i+1}-1]}{x_{j}^{i-1}}+p_{j}
\end{equation}
Then, from~\eqref{eq:pk1'} and~\eqref{eq:sx}, we can derive that $p_{k+1}=\sum_{i=0}^{2^k-1}{x_i^{0}}+\sum_{i=1}^{k}{p_i}=\sum_{i=0}^{2^k-1}{x_i}+\sum_{i=1}^{k}{p_i}$ which satisfies~\eqref{eq:PI}. Thus, we can obtain the correct codeword via this approach.

Algorithm~\ref{alg:2} depicts the explicit steps. In this algorithm, Lines 1-3 handle the basis case $k=1$, namely $x_{0}^{k-1}$ and $x_{1}^{k-1}$ for $p_{k+1}$ and $p_{k}$ respectively. Lines 4-5 calculate $S$ according to the adjustment. Lines 6-9 calculate $x_{i}'$ referred to in \eqref{eq:49}. Then Line 10 calls the procedure recursively on $\mathbf{x}'$. Line 11 presents the return value, which is parity vector. Notably, when Algorithm~\ref{alg:2} is used in the encoding of extended Hamming codes, the input $\mathbf{x}$ shall satisfy the constraint that $x_0=0$ and $x_{2^i}=0,\forall i\in [0, k-1]$. Figure \ref{fig:fig2} presents the circuit of the proposed Algorithm \ref{alg:2} at $k=3$.

\begin{figure}
\center
   \includegraphics[width=0.7\columnwidth]{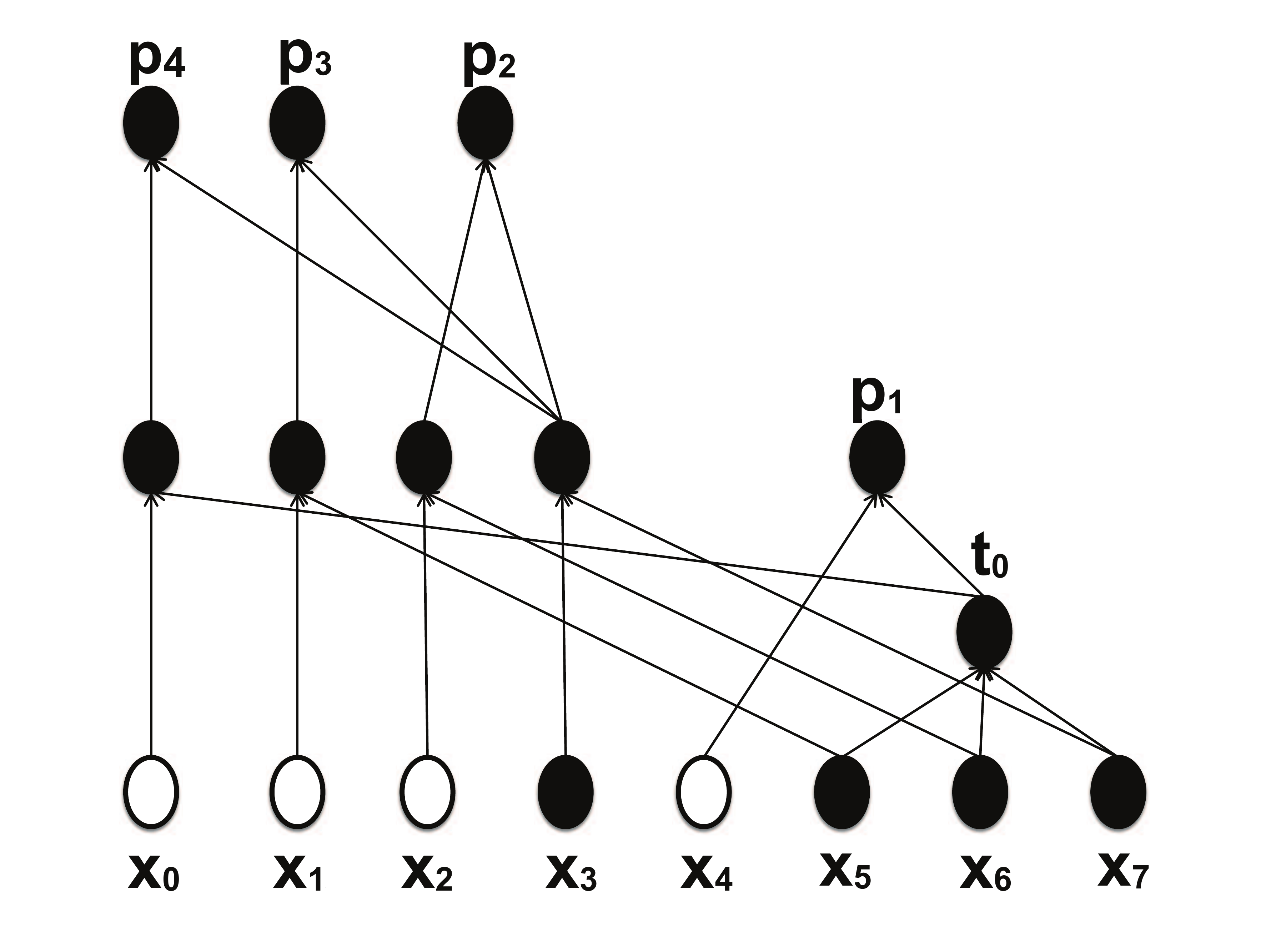}
\caption{\label{fig:fig2} Circuit for Algorithm \ref{alg:2} with $k=3$}
\end{figure}

\begin{algorithm}[t]
\caption{$\mathrm{P}_{4}(\mathbf{x}, k)$: Encoding of extended Hamming codes (Recursive version)}\label{alg:2}
\KwIn{$\mathbf{x}=[x_0\quad  x_1\quad \dots \quad x_{2^{k}-1}]$, and $k$}
\KwOut{$\mathbf{p}_k'$}
\If{$k=1$}{
\Return $[x_1 \quad x_0]$
}
$\alpha \leftarrow \sum_{i=1}^{2^{r-1}-1}x_{2^{r-1}+i}$\\
$S \leftarrow x_{2^{r-1}}+\alpha$\\
$x_{0}'=x_{0}+\alpha$\\
\For{$i= 1$ \KwTo $(2^{k-1}-1)$}{
$x_i'\leftarrow x_i+ x_{2^{k-1}+i}$
}
Call $\mathbf{p}'\leftarrow \mathrm{P}_{4}(\mathbf{x}', k-1)$, where $\mathbf{x}'=[x_0'\quad x_1'\quad \dots\quad x_{2^{k-1}-1}']$\\
\Return $[S \quad \mathbf{p}']$
\end{algorithm}

\subsection{Circuit sizes and parallel depths}\label{sec:com}
The number of XORs for $[2^{k}-1, 2^{k}-k-1]$ Hamming codes (Algorithm~\ref{alg:1}) and for $[2^{k}, 2^{k}-k-1]$ extended Hamming codes (Algorithm~\ref{alg:2}) are denoted as $A_3(k)$ and $A_{4}(k)$, respectively. 
In Algorithm \ref{alg:1}, Lines 1-3 give the base case that
\begin{equation}\label{eq:a2s}
A_3(1)=0.
\end{equation}
Line 4 requires $2^{k-1}-1$ additions, and Lines 5-7 require $2^{k-1}$ XORs. Line 9 calls the procedure recursively. In summary, the recurrence relations are written as
\begin{equation}\label{eq:a2r}
\begin{aligned}
A_3(k)&=A_3(k-1)+2^{k}-1.
\end{aligned}
\end{equation}
The solution is given by
\begin{equation}\label{eq:a2rr}
A_3(k)=2^{k+1}-k-3.
\end{equation}

Then, for Algorithm \ref{alg:2}, $A_4(k)$ has the same recursive formula and the same recursion depth. Thus $A_4(k)=A_3(k)$
\begin{equation}\label{eq:31}
A_4(k)=2^{k+1}-k-3.
\end{equation}

%
Moreover, there are some redundant XOR operations in Algorithms \ref{alg:1} and \ref{alg:2} that can be eliminated. First, the zero-XORs $a+0=a$ can be eliminated in the algorithms. In the input $\mathbf{X}=[x_{0} \dots x_{2^k-1}]$, we have $x_{i}=0$ for $i=0$ and $i$ a power of two. This gives $k+1$ zero-XORs in the second layer of the circuit. Second, for Algorithm \ref{alg:1}, as shown in Fig. \ref{fig:fig1}, the left node of the top layer is redundant, and this leads that the left node in the prior layer is also redundant. By inductions, it can be shown that the left node in every layer can be removed. In summary, a total of $k-2$ XORs can be eliminated from the third layer to the top layer~(the $k$-th layer). Thus, $A_3(k)$ and $A_4(k)$ can be reduced to 
\begin{equation}\label{reduced-A_3}
2^{k+1}-3k-2
\end{equation}
and 
\begin{equation}\label{reduced-A_4}
2^{k+1}-2k-4,
\end{equation}
respectively. By \eqref{reduced-A_3} and  \eqref{reduced-A_3} given in Section \ref{sec:3}, the circuits sizes analysis shows that the $[2^{k}-1, 2^{k}-k-1]$ Hamming code exactly requires $2^{k+1}-3k-2$ XORs by Algorithm~\ref{alg:1}, and the $[2^{k}, 2^{k}-k-1]$ extended Hamming code exactly requires $2^{k+1}-2k-4$ XORs by Algorithm \ref{alg:2}. This concludes that the lower bounds of circuit sizes shown in Theorem \ref{th:2} and Theorem \ref{th:3} are achievable. As shown before, these proposed encoding algorithms achieve the minimal circuit sizes. Note that the parallel depth of the circuit in Figures \ref{fig:fig1} and \ref{fig:fig2} are both $2$. Thus, it is easy to verify that the parallel depth of Algorithms \ref{alg:1} and \ref{alg:2} are both $k-1$.

\subsection{Shortened Hamming codes}\label{sec:5D}
In this section, we generalize the shortened Hamming codes in Section \ref{sec:2E} to the encoding algorithm of the $[2^{k-1}+k, 2^{k-1}]$ shortened Hamming codes $\mathcal{C}$. Further, the circuit size and the parallel depth are analyzed.

For a codeword $\mathbf{c}=[x_{0} \dots x_{2^{k-1}-1} \ p_{0} \dots p_{k-1}]$ of $\mathcal{C}$, its parity check matrix is given by
\begin{equation}\label{eq:64}
\begin{bmatrix}
T^{k}(2^{k-1}-1) & T^{k}(2^{k-1}+1) & T^{k}(2^{k-1}+2) & \dots & T^{k}(2^{k}-1) & T^{k}(2^{0}) & T^{k}(2^{1}) & \dots & T^{k}(2^{k-1})
\end{bmatrix}.
\end{equation}
One can see that \eqref{eq:64} is a submatrix of $\mathbf{H}_{k}$ defined in \eqref{equ:2}. This shows that $\mathcal{C}$ is a class of $[2^{k-1}+k, 2^{k-1}]$ shorten Hamming codes. From Table \ref{tab:2}, it is easy to verify that the shorten Hamming codes \cite{Warren:2002:HD:515297} is $\mathcal{C}$ with $6$ parity bits and $32$ message bits.

Then we discuss the encoding algorithm of $\mathcal{C}$. The approach \cite{Warren:2002:HD:515297} is presented in the software implementation. To show the encoding complexity in terms of circuit size an parallel depth, we describe the algorithm with the encoding circuit of $k=8$ message bits in Figure \ref{fig:fig3}. The figure shows the procedure of obtaining the set of outputs $\{p_{i}'\}_{i=0}^3$, and the parities are given by $p_{i}=p_{i}'+x_{0}$.
Let $A_{5}(k)$ denote the circuit size of the algorithm \cite{Warren:2002:HD:515297} with $2^{k}$ message bits. As shown in Figure \ref{fig:fig3}, calculating $p_{k}'$ requires $\sum_{i=0}^{k-1}{2^{i}}=2^{k}-1$ XORs to build the complete binary tree. Further, calculating $\{p_{i}'\}_{i=0}^{k-1}$ requires a total of $\sum_{i=1}^{k-1}{(2^{i}-1)}=2^{k}-k-1$ XORs, and calculating $\{p_{i}=p_{i}'+x_{0}\}_{i=0}^k$ requires $k+1$ XORs. In summary, $A_{5}(k)=2^{k+1}-1$. Moreover, from Figure \ref{fig:fig3} and $p_{i}=p_{i}'+x_{0}$, the parallel depth of this algorithm with message length $2^{k}$ is $k+1$.


The above shows that the shortened Hamming codes \cite{Warren:2002:HD:515297} requires around $2$ XORs per message bit. However, we do not aware any literature to prove the computational lower bound of the shortened Hamming code. This causes that we do not know the gap between the algorithm \cite{Warren:2002:HD:515297} and the real lower bound.

\begin{figure}
	\center
	\includegraphics[width=0.7\columnwidth]{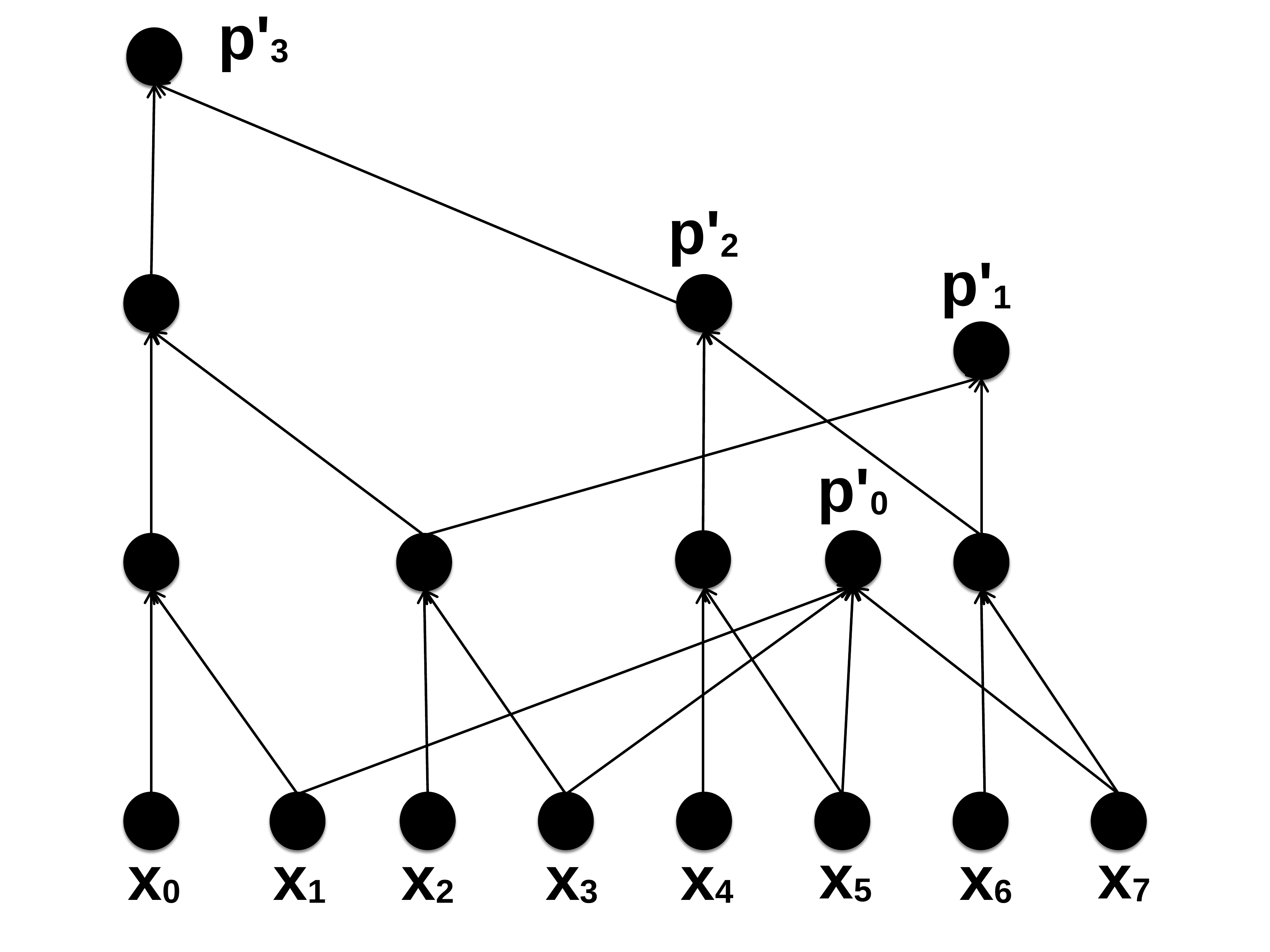}
	\caption{\label{fig:fig3} Part of the circuit of the algorithm \cite{Warren:2002:HD:515297} with $8$ message bits}
\end{figure}

%
%

\section{Conclusions}\label{sec:6}
In this paper, the encoding circuit sizes of (extended)~Hamming codes and (punctured)~Hadamard codes are investigated. The exact lower bounds on the encoding circuit sizes of these codes are presented. Further, we also proposed encoding algorithms for (punctured)~Hadamard codes and (extended)~Hamming codes that reach the exact lower bounds of circuit size, to show that these lower bounds are achievable. To our knowledge, we are the first to show a exact and tight lower bound on encoding computational complexity of non-trivial linear block codes. A possible future work is to find the exact lower bounds on the encoding circuit size of other more complex linear error correcting codes.


\bibliographystyle{spmpsci}      
\bibliography{Hamming_codes}   


\end{document}